\documentclass{amsart}
\usepackage[T1]{fontenc}
\usepackage[utf8]{inputenc}
\usepackage{todonotes}
\usepackage{graphicx}
\usepackage[style=numeric]{biblatex}
\newtheorem{theorem}{Theorem}
\newtheorem{proposition}{Proposition}
\newtheorem{corollary}{Corollary}[theorem]
\newtheorem{lemma}{Lemma}
\theoremstyle{definition}
\newtheorem{definition}{Definition}
\theoremstyle{remark}
\newtheorem{example}{Example}

\newcommand{\trsp}[1]{#1^\mathsf{T}} 

\def\ba{\mathbf a}
\def\bb{\mathbf b}

\usepackage{fancyvrb}
\fvset{fontshape=normal}
\fvset{listparameters={\setlength{\topsep}{0pt}\setlength{\partopsep}{0pt}}}

\usepackage{graphicx}

\usepackage[ruled]{algorithm2e}
\SetKwProg{For}{for}{}{}
\SetKwProg{If}{if}{}{}
\usepackage{hyperref}
\hypersetup{colorlinks=true}
\usepackage{lineno}
\usepackage[capitalize]{cleveref}
\definecolor{darkred}{rgb}{0.5, 0, 0}

\title{Positivity Proofs for Linear Recurrences through Contracted
Cones}  
\author{Alaa Ibrahim\and Bruno Salvy}
\email{Alaa.Ibrahim@inria.fr, Bruno.Salvy@inria.fr}
\address{Univ Lyon, EnsL, UCBL, CNRS, Inria,  LIP, F-69342, LYON Cedex 07, France}

\date{\today}

\begin{document}
\begin{abstract}
Deciding the positivity of a sequence defined by a linear recurrence with polynomial coefficients and initial condition is difficult in general. Even in the case of recurrences with constant coefficients, it is known to be decidable only for order up to~5. We consider a large class of linear recurrences of arbitrary order, with polynomial coefficients, for which an algorithm decides positivity for initial conditions outside of a hyperplane. The underlying algorithm constructs a cone, contracted by the recurrence operator, that allows a proof of positivity by induction. The existence and construction of such cones relies on the extension of the classical Perron-Frobenius theory to matrices leaving a cone invariant.
\end{abstract}

\maketitle

\section{Introduction}
A sequence $(u_n)_{n\in \mathbb{N}}$ of real numbers is called \emph{P-finite} if it satisfies a  linear recurrence 
 \begin{equation}\label{rec}
    p_d(n)u_{n+d}=p_{d-1}(n) u_{n+d-1}+\dots+p_0(n)u_n,\qquad n\in \mathbb{N},
 \end{equation}
with coefficients $p_i\in\mathbb R[n]$ and~$p_d\neq0$. The sequence is moreover called \emph{C-finite} when the polynomials~$p_i$ are constant.
The problem of positivity is to determine whether $u_n\ge0$ for all $n\in\mathbb N$ or not.

When $p_d(n)\neq0$, \cref{rec} defines~$u_{n+d}$ in terms of the previous ones. 
As~$p_d$ can only have finitely many zeros, up to checking $n_0$ elements of the sequence first and considering the shifted sequence~$v_n=u_{n+n_0}$, one can assume without loss of generality that $p_d(n)\neq0$ for all~$n\in\mathbb N$ and thus that the sequence is entirely determined by its initial conditions~$u_0,\dots,u_{d-1}$. 
By the same argument one can also assume that $p_0(n)\neq0$ for
all~$n\in\mathbb N$, which is a hypothesis in our 
\cref{thm:decidability}.

For reasons of effectivity, it is necessary to restrict to a subfield of~$\mathbb R$. For simplicity, we state our algorithms and results over the rationals; real number fields can be handled similarly. Thus, we focus on the following.
\begin{definition}[Positivity problem for P-finite sequences] Given~$p_0,\dots,p_d$ in $\mathbb Q[n]$ such that $0\not\in p_0p_d(\mathbb N)$ and given $u_0,\dots,u_{d-1}$ in~$\mathbb Q$, decide whether the sequence $(u_n)$ defined by \cref{rec} satisfies $u_n\ge0$ for all~$n\ge0$.
\end{definition}

P-finite and C-finite sequences are closed under addition, product, Cauchy product $((u_n),(v_n))\mapsto(\sum_{k=0}^nu_kv_{n-k})$ and the operation of taking subsequences $(u_{\ell n+q})$ for fixed $\ell\in\mathbb N_{>0}$ and $q\in\mathbb N$. These operations are all effective: algorithms produce recurrences for these sequences given recurrences for the input~\cite[\S6.4]{Stanley1999}. This implies that if $(u_n)$ is a P-finite sequence, proving monotonicity ($u_{n+1}\ge u_n$), convexity ($u_{n+1}+u_{n-1}\ge 2u_n$), log-convexity ($u_{n+1}u_{n-1}\ge u_n^2$) or more generally an inequality with another P-finite sequence ($u_n\ge v_n$) are all problems that reduce to the positivity problem for P-finite sequences.

For applications of the positivity problem of C-finite sequences, we refer to the numerous references in the work of Ouaknine and Worrell~\cite{OuaknineWorrell2014a}. Motivations for studying positivity in the more general context of P-finite sequences also come from various areas of mathematics and its applications, including number theory~\cite{StraubZudilin2015}, combinatorics~\cite{ScottSokal2014}, special function theory~\cite{Pillwein2008}, or even biology~\cite{MelczerMezzarobba2022,YuChen2022,BostanYurkevich2022a}. In computer science, the verification of loops allowing multiplication by the loop counter leads to P-finite sequences~\cite{HumenbergerJaroschekKovacs2017,HumenbergerJaroschekKovacs2018}. Positivity questions for such recurrences also occur in the floating-point error analysis of simple loops obtained by discretization of linear differential equations~\cite{BoldoClementFilliatreMayeroMelquiondWeis2014} and in the numerical stability of the computation of sums of convergent power series~\cite{SerraArzelierJoldesLasserreRondepierreSalvy2016}. 

Several examples come from \emph{diagonals} of multivariate rational functions (see for instance~\cite{Christol2015}). This is the case of the sequence 
\begin{equation}\label{SZ}
s_n=\sum_{k=0}^n{(-27)^{n-k}2^{2k-n}\frac{(3k)!}{k!^3}\binom{k}{n-k}},
\end{equation}
whose positivity was proved by Straub and Zudilin~\cite{StraubZudilin2015} using properties of hypergeometric series, but that can also be proved from the recurrence
\[2(n+2)^2s_{n+2}=(81n^2+243n+186)s_{n+1}-81(3n+2)(3n+4)s_n,\]
although not in an immediate way. Also, a family of diagonals comes from the sequences 
\begin{equation}\label{eq:GRZ}
u_n^{(k)}=\sum_{j=0}^n(-1)^j\frac{(kn-(k-1)j)!k!^j}{(n-j)!^kj!}
\end{equation}
whose positivity for $k\ge4$ was a long-standing conjecture by Gillis, Reznick and Zeilberger~\cite{GillisReznickZeilberger1983}, proved only recently by Yu~\cite{Yu2019}. In our setting, this gives a family of tests of increasing size, as the sequences~$(u_n^{(k)})$ satisfy linear recurrences. Empirically, the $k$th one has order~$k$ and coefficients of degree~$k(k-1)/2$, see also~\cite{Pillwein2019}.

\subsection*{Decidability Issues}
Even in the case of C-finite sequences (when the recurrences have constant coefficients), the positivity problem is not completely understood. If 
\[\chi(x)=x^d-\frac{p_{d-1}}{p_d}x^{d-1}-\dots-\frac{p_0}{p_d}\]
is the characteristic polynomial of the recurrence and $\lambda_1,\dots,\lambda_k$ are its distinct roots, then it is classical that the solution takes the general form
\[u_n=C_1(n)\lambda_1^n+\dots+C_k(n)\lambda_k^n,\]
with polynomial $C_i$ that can be computed from the initial conditions~$u_0,\dots,u_d$.
Still, the positivity problem is difficult in this situation. It is only known to be decidable for order~$d\le 5$ and the case $d=6$ is related to open problems in Diophantine approximation~\cite{OuaknineWorrell2014a}. Moreover, by the closure properties for sum and product, the famous \emph{Skolem problem}, which asks whether a C-finite sequence has a~0 or not and is only known to be decidable for order~$d\le 4$, reduces to the positivity problem (in higher order)~\cite{OuaknineWorrell2014a}. For special classes of C-finite recurrences (simple, or reversible), positivity has been proved decidable for slightly larger order~\cite{LiptonLucaNieuwveldOuakninePurserWorrell2022,Kenison2022a,KenisonNieuwveldOuaknineWorrell2023}.
There is one subclass of C-finite recurrences of arbitrary order for which deciding the sign is easy: it is when one of the roots~$\lambda_i$, say~$\lambda_1$, dominates the other ones, in the sense that $|\lambda_1|>|\lambda_i|$ for $i\neq1$ and moreover the corresponding polynomial~$C_1(n)$ is not zero. Then asymptotically the term~$C_1(n)\lambda_1^n$ dominates~$u_n$ and one can bound effectively the other ones~\cite{OuaknineWorrell2014b}. The work presented here is a generalization of this favorable situation in the case of polynomial coefficients.

\subsection*{Previous Work}
In computer algebra, an important method was developed by Gerhold and Kauers~\cite{Gerhold2005,GerholdKauers2005}. The principle is to look for~$m$ such that the formula
\begin{multline*}
\forall n\ge0,\forall u_n\ge0,\forall u_{n+1}\ge0,\dots,\forall u_{n+d-1}\ge0,\\
u_{n+d}\ge0\wedge\dotsb\wedge u_{n+m}\ge0\Rightarrow u_{n+m+1}\ge0
\end{multline*}
can be proved by quantifier elimination and to increase~$m$ until the method succeeds or a prescribed bound has been reached. While there is no guarantee of success in general~\cite{KauersPillwein2010a}, this method is very powerful, giving for instance an automatic proof of Turán's inequality for Legendre polynomials~\cite{GerholdKauers2006}. Kauers and Pillwein refined this method for the special case of P-finite sequences. In particular, they observed that it is easier to prove the stronger property~$\exists\mu>0,u_{i+1}\ge\mu u_i$ and proved termination of the method for cases with a unique dominant eigenvalue (see \cref{def:dominant-eigenvalues}) and order~$d=2$, under a genericity assumption on the initial conditions; they also extended their analysis to a subclass of recurrences of order~$d=3$~\cite{Kauers2007a,KauersPillwein2010a,Pillwein2013}. For an appropriate choice of~$\mu$, this is also sufficient to prove the positivity of the sequences \eqref{eq:GRZ} for $k=4,\dots,17$~\cite{Pillwein2019}. Our approach can be viewed as extending this idea of introducing extra inequalities (the faces of the cones it constructs) to make a proof by induction feasible in arbitrary order. In particular, we recover termination in all cases for which this was proved before. Note that for order~2, a different approach based on continued fractions is also possible~\cite{KenisonKlurmanLefaucheuxLucaMoreeOuaknineWhitelandWorrell2021}.

\subsection*{Contribution}

It is classical that the recurrence from \cref{rec} can also be viewed
as a first-order linear recurrence~$U_{n+1}=A(n)U_n$ on vectors
in~$\mathbb R^d$, where $U_n=\trsp{(u_n,\dots,u_{n+d-1})}$ (see 
\cref{sec:eigenvals}) and $A(n)$ is a companion matrix. We consider
the equation~$U_{n+1}=A(n)U_n$ for arbitrary matrices~$A(n)$ of
rational functions. This offers a more geometric view on the positivity
problem: the vector~$U_n$ must remain in the cone~$\mathbb R_{\ge0}^d$. In the case of a constant matrix~$A$, there is a well-developed Perron-Frobenius theory for cones~\cite{TamSchneider2014} that relates spectral properties of the matrix~$A$ to the existence of cones that are stable under~$A$ in the sense that $AK\subset K$, or that are contracted ($AK\subset K^\circ$) (the classical situation is~$K=\mathbb R_{\ge0}^d$.) For the large class of recurrences of \emph{Poincaré type} (see \cref{def:poincare}), the matrix~$A(n)$ can be viewed as a perturbation of its limit~$A$ as~$n\rightarrow\infty$. This allows us to construct a subcone of $\mathbb R_{\ge0}^d$ that is stable under~$A(n)$ for $n$ sufficiently large and then one only has to check that the initial terms of the sequence are positive. This last step is always possible under a genericity condition that we make explicit, but that is not effective.
\begin{algorithm}
\caption{Main Steps of the Positivity Algorithm}
\label{algorithm:General}
\KwIn{ $A(n)\in\mathbb Q(n)^{d\times d}$ such that 
\( A = \lim\limits_{n \to \infty} A(n)\in\mathbb{Q}^{d\times d} \);
 $U_0\in\mathbb Q^d$ }
\KwOut{Positivity of $(U_n)_n$}
\begin{enumerate}
  \item[1.]  \textbf{Cone Construction:} Construct a cone
  $K\subset\mathbb R_{\ge0}^{d}$ s.t. $A
  (K\setminus \{0\})
\subset K^{\circ}$ 
  \item[2.] \textbf{Stability index:} Compute $m\in \mathbb{N}$
  s.t. for all $n\geq m$, $A(n)K\subset K$
   \item[3.] \textbf{Initial conditions check:}
 Compute $n_0\geq m$ s.t. $U_{n_0}\in K$ and 
 check if $U_0,U_1,\dots,U_{n_0-1}$ are positive.
\end{enumerate}
\end{algorithm}
An overview of the algorithm is given in \cref{algorithm:General}.

\begin{example}\label{example:3}
\begin{figure}
    \centering
    \includegraphics[width=0.25\textwidth]{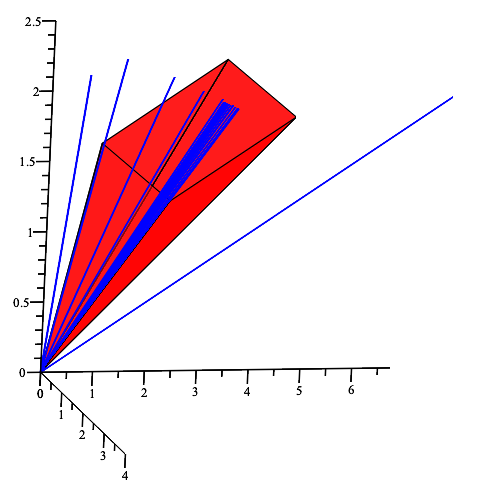}
    \caption{The first values of $U_n=(u_n,u_{n+1},u_{n+2})$ (in blue) in \cref{example:3}, together with the corresponding cone $K$ (in red).}
    \label{fig:example}
\end{figure}
\Cref{fig:example} displays the first values of the sequence
$(u_n)_n$ defined by $u_0=1,u_1=3,u_2=1$ and
    \[(16n+1)u_{n+3}=(32n-2)u_{n+2}-(21n-4)u_{n+1}+(5n-3)u_{n},\qquad n\ge0.\]
In this example, Step~1 constructs the red cone, Step~2 shows that for
$n\ge6$, this cone is stabilized by the matrix (in the sense that $A
(n)K\subset K$) and finally Step~3 checks that $U_6=\trsp{(u_6,u_7,u_8)}\in
K$ and that $u_0,u_1,\dots,u_5$ are positive. These conditions prove that \( (u_n)_n \) is a positive sequence.
 \end{example}

Our main result is the following theorem.

\begin{theorem}\label{thm:decidability}
For all linear recurrences of the form given in \cref{rec}, of order $d$ and of Poincar\'e
type, having a unique simple dominant eigenvalue and such that
$0\not\in p_0p_d(\mathbb N)$, the positivity of the solution
$(u_n)_n$ is decidable for any $U_0=(u_0,u_1,\dots,u_{d-1})$ outside
a hyperplane in $\mathbb{R}^d$.
\end{theorem}

A first algorithm giving this result was presented at the conference 
\textsc{Soda}~\cite{IbrahimSalvy2024}. The cones used in that work
were obtained from a pseudo-metric on~$\mathbb R_{>0}^d$ due to
Hilbert together with recent work of Friedland~\cite{Friedland2006}.
While this approach gives the same decidability result as \cref{thm:decidability}, the algorithm underlying it has severe inefficiencies on large examples, in particular because it requires to work with a power of the matrix instead of the matrix itself. The work presented here relies on cones that are simpler to handle and capture more closely the geometry of the iteration.

\subsection*{Positivity vs Non-negativity} We follow the literature in the definition of the positivity problem above, calling `positive' sequences that satisfy $u_n\ge0$ rather than $u_n>0$. For the class of recurrences to which our results apply, there is no difference in the difficulty of proving either condition. In terms of vocabulary, we call a vector~$V$ \emph{positive} (resp. non-negative) and write~$V>0$ (resp. $V\ge0$) when all its entries are positive (resp. non-negative), meaning that they belong to~$\mathbb R_{>0}$ (resp.~$\mathbb R_{\ge0}$). 

\subsection*{Structure of the Article} The basic definitions and
properties related to the matrix version of the recurrence are
introduced in \cref{sec:eigenvals}, those related to cones are given in 
\cref{sec:cones}. The proof of a more precise version of 
\cref{algorithm:General} and of \cref{thm:decidability} are given in 
\cref{sec:proof}. Several possible choices of cones for the algorithm
are given in \cref{section:constructions} and the related algorithms
in the following one. Possible improvements in the specific case of
matrices originating from linear recurrences like \cref{rec} are
discussed in \cref{sec:linrec}. A detailed example is treated in \cref{sec:example}.
\Cref{section:experiments} presents experimental results and perspectives are discussed in the final section.

\section{Matrix Recurrences}\label{sec:eigenvals}
\subsection{Scalars to Vectors}
A more geometric view on the recurrence relation in \cref{rec} is obtained by converting it into a vector form.
If $U_n$ denotes the vector $\trsp{(u_{n},\dots,u_{n+d-1})}$, then
\cref{rec} gives rise to the first-order linear recurrence 
\begin{equation}\label{eq:rec-vector}
U_{n+1}=A(n)U_n,
\end{equation}
where $A(n)\in\mathbb Q(n)^{d\times d}$ is the companion matrix
\begin{equation}\label{eq:companion}
A(n)=\begin{pmatrix}
    0 & 1 & 0 & \dots& 0\\
    0 & 0 & 1 & \dots& 0\\
    \dots&\dots &\dots&\dots&\dots\\
    0 &0  &0 &\dots&1\\
    \frac{p_0(n)}{p_d(n)}&\frac{p_1(n)}{p_d(n)} & \frac{p_2(n)}{p_d(n)}&\dots&\frac{p_{d-1}(n)}{p_d(n)}
    \end{pmatrix}, \qquad n\in\mathbb{N}.
\end{equation}
The sequence $U_n$ is then recovered from the vector of initial conditions by the matrix factorial \[U_n=A(n-1)A(n-2)\dotsb A(0)U_0.\]
The positivity of the sequence $(u_n)_{n\in\mathbb N}$ becomes
equivalent to $U_n\in{\mathbb R}_{\ge0}^d$ for all $n\in\mathbb N$.

\subsection{Constant Matrices and the Power Method}
Let $A$ be a matrix in~$\mathbb C^d$. 
The power method is a classical method to compute an eigenvector of such a matrix. The idea is to pick a random vector~$U_0$ and study the sequence $U_{n+1}=AU_n/\|U_n\|$, that, under suitable conditions, converges to an eigenvector.
This convergence is well-studied~\cite{ParlettPoole1973} and we focus here on a simple but generally satisfied sufficient condition given in \cref{lemma:power-method} below in terms of \emph{dominant eigenvalues}.
\begin{definition}[Dominant eigenvalues]\label{def:dominant-eigenvalues}
Let
$\lambda_1,\dots,\lambda_m$ be the distinct complex eigenvalues of the matrix $A$, numbered by decreasing modulus so that 
\[|\lambda_1|=|\lambda_2|=\dots=|\lambda_k|>|\lambda_{k+1}|\geq |\lambda_{k+2}|\dots \geq |\lambda_m|.\]
Then $\lambda_1,\dots,\lambda_k$ are called the \emph{dominant eigenvalues of $A$} (or equivalently \emph{dominant roots} of its characteristic polynomial). An eigenvalue is called \emph{simple} when it is a simple root of the characteristic polynomial.
\end{definition}
Given a characteristic polynomial in~$\mathbb Q[x]$, determining that it has a simple dominant root can be achieved in polynomial bit
complexity~%
\cite{GourdonSalvy1996,BugeaudDujellaFangPejkovicSalvy2022}. The following lemma follows from~\cite{ParlettPoole1973}; it is also a special case of Friedland's result given in \cref{thm:Friedland} below.

\begin{lemma}\label{lemma:power-method}
Let $A$ be a complex matrix with a simple dominant eigenvalue~$\lambda_1$. Then there exist nonzero vectors~$V$ and~$W$ with $AV=\lambda_1V$, $\trsp{A}W=\lambda_1W$ and a sequence $(\epsilon_n)$ of complex numbers with $|\epsilon_n|=1$ for all~$n$ such that 
\[\lim_{n\rightarrow\infty}\epsilon_n\frac{A^n}{\|A^n\|}\rightarrow V\trsp{W}.\]
\end{lemma}
Thus for a \emph{generic} $U_0$ in the sense that $\trsp{W}U_0\neq0$, the direction of the sequence~$(U_n)$ converges to that of the eigenvector~$V$. 

This gives the geometric basis for a positivity algorithm in the case of linear recurrences with constant coefficients whose characteristic polynomial admits a unique dominant root. The general situation will then be analyzed as a perturbation of this situation.

\subsection{Recurrences of Poincaré Type}
An important subclass of linear recurrences whose behaviour is similar to that of recurrences with constant coefficients is the following.
\begin{definition}\label{def:poincare}
A linear recurrence of the form given in~\cref{rec} is said to be of 
\emph{Poincaré type} if the matrix \( A := \lim_{n \rightarrow \infty}
A(n) \) is finite and different from~0, i.e., all entries of \( A(n) \) converge to a
finite limit and at least one of them is not~0.
\end{definition}
The asymptotic behaviour of these recurrences is then related to the dominant eigenvalues of the limit matrix~$A$. 
(The condition that the limit is nonzero is only relevant when the order~$d$ is~1, where the companion matrix is a $1\times1$ matrix whose only entry is a rational function.)

It is always possible to reduce the positivity problem of a P-finite sequence to that of the solution of a linear recurrence of Poincaré type by re-scaling~\cite[\S2]{MezzarobbaSalvy2010}.
\begin{example}
The recurrence 
\[u_{n+3}+u_{n+2}+nu_{n+1}+(n+1)u_n=0\]
is not of  Poincaré type: the degrees of the last two coefficients are larger than that of the first one. The sequence defined by the recurrence
\[(n+2)\psi_{n+2}=\psi_n,\quad \psi_0=\psi_1=1\]
is positive. Thus the positivity of $(u_n)$ is equivalent to that of the sequence $(v_n)$ defined by $v_n=u_n\psi_n$. By closure under product, the sequence $(v_n)$ is P-finite. It satisfies the recurrence
\[
(n+6)(n+4)v_{n+6}+2(n+4)(n+1)v_{n+4}
+(n^2-n-5)v_{n+2}-(n+1)v_n=0,
\]
of Poincaré type.
\end{example}
The same operation can also be used if the matrix $A$ is nilpotent, using a recurrence of the form $u_{n+q} = n^pu_n$ instead, so that one can always assume that the recurrence is of Poincar\'e type, with $A$ having a nonzero eigenvalue.

\section{Cones}\label{sec:cones}
Our approach relies on the Perron-Frobenius theory for cones. 
We first recall basic definitions and properties of cones. We refer to the survey \cite{TamSchneider2014} for more information, including references to earlier work.

\subsection{Definitions}

\subsubsection*{Cone}
A subset \( K \) of \( \mathbb{R}^d \) is called a \emph{cone} if it is closed under addition and multiplication by positive scalars.

\subsubsection*{Proper cone} 
A cone \( K \) in \( \mathbb{R}^d \) is called \emph{proper} if it satisfies the following properties:
\begin{itemize}
    \item[--] \( K \) is \emph{pointed}, i.e., \( K \cap (-K) = \{0\} \).
    \item[--] \( K \) is \emph{solid}, i.e., its interior \( K^{\circ} \) is non-empty.
    \item[--] \( K \) is \emph{closed} in \( \mathbb{R}^d \).
\end{itemize}
All cones considered in this article are proper.
\subsubsection*{Extremal vectors}
A vector $x\in K$ is called \emph{extremal} if $x=y+z$ with $y\in K$ and $z \in K$ implies that both $y$ and $z$ are nonnegative multiples of $x$.

An important property is that a cone $K\subset\mathbb{R}^d$ is generated by its extremal vectors: any element of $K$ can be written as a finite linear combination of its extremal vectors with nonnegative scalar coefficients (see, e.g., \cite{Vandergraft1968}).

\subsubsection*{Polyhedral cone}
A cone $K\subset \mathbb{R}^d$ is called \emph{polyhedral} if it has a finite number of extremal vectors.

\subsubsection*{Positive cone}
A cone \( K \) is called \emph{positive} if  \( K\subset\mathbb{R}_{\ge0}^d \).

\subsubsection*{Contracted cone}
A cone $K$ is \emph{contracted} by a real matrix $A$ if $A(K\setminus \{0\})
\subset K^{\circ}$, where $K^\circ$ denotes the interior of $K$.

\subsection{Properties}
The study of the relation between dominant eigenvalues and contracted cones is an important part of the generalization of the Perron-Frobenius theory to cones. The main result we need is the following.

\begin{proposition}[Vandergraft \cite{Vandergraft1968}]
\label{thm:Vandergraft}
Let \( A \) be a real matrix. There exists a proper cone \( K \) contracted by $A$ if and only if \( A \) has a unique dominant eigenvalue \( \lambda > 0 \) that is simple.
The interior $K^\circ$ of \( K \) contains a unique eigenvector of \( A \) associated with \( \lambda \), up to constant multiples.
\end{proposition}
Vandergraft's proof is constructive. Variants can also be
turned into algorithms, presented in \cref{section:constructions}.

\section{Proof of the Algorithm}\label{sec:proof}
We now prove \cref{thm:decidability} by 
first proving the correctness of \cref{algorithm:PP}.
 \begin{algorithm}[t]
\caption{Positivity algorithm for matrix recurrences}
\label{algorithm:PP}
\KwIn{ $A(n)\in\mathbb Q(n)^{d\times d}$ ;
 $U_0\in\mathbb Q^d$\\ $A=\lim_{n\rightarrow\infty}A(n)$ has a unique simple dominant eigenvalue $\lambda>0$\\ and a positive eigenvector associated to it}
\KwOut{One of Positive, Non-Positive}

 1.\hspace{0.5cm}Construct a cone $K\subset\mathbb{R}^d_{\ge0}$ s.t $A
 (K\setminus \{0\})\subset K^{\circ}$ \\
  2.\hspace{0.5cm}Compute $m\in \mathbb{N}$ such that for all $n\geq m$, $A(n)K\subset K$\\

  3.\hspace*{0.5cm}{\bf for }{$i=0,\dots,m-1$ }{\bf if }{{$U_i\not\ge0$ }{\bf then }{\bf return }(Non-positive)}\\
    4.\hspace*{0.5cm}{\bf for }{$i=m,m+1,\dots,\infty$}\\
    \hspace*{1.5cm}{\bf if }{$U_i\not\ge0$ }{\bf then }{\bf return } Non-positive\\
    \hspace*{1.5 cm}{\bf if }{$U_i \in K$ }{\bf then}{ \bf return } Positive
\end{algorithm}

\begin{theorem}\label{theorem:main}
Let~$A(n)$ be a matrix in $\mathbb Q(n)^{d\times d}$, invertible for all $n\in\mathbb N$, and tending to a finite limit~$A$ as~$n\rightarrow\infty$, that has a unique simple positive dominant eigenvalue and a positive eigenvector associated to it. Then there exists a vector $W\in\mathbb R^d$ such that  positivity of the solution of~$U_{n+1}=A(n)U_n$ given~$A(n)$ and~$U_0\in\mathbb Q^d$ can be decided when $\trsp{W}U_0\neq0$.
\Cref{algorithm:PP} either disproves positivity or proves it in
the generic situation $\trsp{W}U_0\neq0$.
\end{theorem}

If $A(n)$ is invertible in $\mathbb Q(n)^{d\times d}$, then its determinant can have only finitely many zeros. Thus in that case, up to checking $n_0$ elements of the sequence first and considering the shifted sequence~$V_n=U_{n+n_0}$, one can assume without loss of generality that $A(n)$ is invertible for all~$n\in\mathbb N$.

\smallskip 
The proof of \cref{theorem:main} is obtained by considering the steps of the algorithm one by one.

\subsection{Cone construction}

\begin{proposition}[Cone with fixed signs]\label{pro:ConeExistence}
Let \( A \in \mathbb{R}^{d \times d} \) be a real  matrix with a
simple dominant eigenvalue \( \lambda > 0 \) and  an
eigenvector $v$ associated to $\lambda$ with nonzero coordinates. Let $\epsilon_i$ be `${\ge0}$' if $v_i>0$ and `$\le 0$' otherwise.
Then there exists a proper cone \( K \subset \mathbb R_{\epsilon_1}\times\dots\times\mathbb{R}_{\epsilon_d} \) contracted by \( A \) containing $v$. In particular, if $v$ is positive, one can take $K\subset\mathbb R_{\ge0}^d$.
\end{proposition}
\begin{proof}
By Vandergraft's result (\cref{thm:Vandergraft}), there exists a
proper cone \( K \subset \mathbb{R}^d \) contracted by $A$ and the
interior of this cone contains a unique dominant eigenvector of \( A
\). Without loss of generality, we can assume \( v \in K^\circ \), as
we could replace \( K \) with \( -K \) if necessary.

We now consider the decreasing sequence of cones
\[
K\supset AK\supset\dots\supset A^m K.
\]
By a result of Tam and Schneider~\cite[Lemma~5.3]{TamSchneider1994}, the
limit of this sequence is
\[
\bigcap_{j=0}^{\infty} A^j K = \{v\}.
\]
Since $v$~has nonzero coordinates and belongs to~$K^\circ$, for $m$ sufficiently large, the cone $A^mK$ is
included in~$\mathbb R_{\epsilon_1}\times\dots\times\mathbb{R}_{\epsilon_d} $.
\end{proof}
A more practical algorithm that avoids repeated multiplication by \( A \) to construct a positive cone  is described in \cref{section:constructions}.

\subsection{Stability index}
The next step of the algorithm requires an index~$m$ such that $A
(n)K\subset K$ for $n\ge m$. A compactness argument gives its existence. \begin{proposition}\label{lem:Contraction}
Let $A(n)$ be in $\mathbb{R}^{d\times d}$ for $n\in\mathbb{N}$ and have finite limit $A$ as $n\rightarrow\infty$. Let $K$ be a proper cone contracted by $A$. Then there exists $m\in\mathbb{N}$  such that $A(n) K\subset K$ for all $n\geq m$.
\end{proposition}
\begin{proof}Fix a norm~$\|\cdot\|$ of $\mathbb R^d$.
Let  $S_{d-1}$ denote the unit sphere in $\mathbb{R}^d$.
By linearity, it suffices to prove that $ A(n) (K \cap S_{d-1}) \subset K $ for $n\ge m$.\\
As $K$ is a closed cone, $K \cap S_{d-1}$ is compact and so is $A(K
\cap S_{d-1})\subset K^\circ$ by continuity of~$A$. It follows that there
exists~$r>0$ such that for all $x\in K
\cap S_{d-1}$, the open ball $B(Ax,r)$ for~$\|\cdot\|$ is included
in~$K$. And then also by compactness, as $ A(n)$ converges to  $A$,
there
exists~$m>0$ such that for all~$n\ge m$ and~$x\in K\cap S_{d-1}$,
$A(n)x\in B(Ax,r)\subset K$, as was to be proved.
\end{proof}

\subsection{Convergence into the Cone}
The previous steps (construction of the cone, stability index) depend
only on the matrix~$A(n)$, but not on the initial conditions. The next
step ensures that the sequence enters the constructed cone, under
genericity assumptions. It relies on the following strong generalization of the classical power method
of \cref{lemma:power-method}.
\begin{proposition} [Friedland~\cite{Friedland2006}]
\label{thm:Friedland}
Let~$A(n)$ be in~$\operatorname{GL}_d(\mathbb R)$ for~$n\in\mathbb N$ and tend to a finite limit~$A$ as $n\rightarrow\infty$, such that $A$ has a unique simple nonzero dominant eigenvalue $\lambda$. 
Then there exist a sequence $(\epsilon_n)$ in $\{-1,1\}^\mathbb N$ and two nonzero vectors $V,W$ s.t. $AV=\lambda V$ and
\[
\lim_{n\rightarrow\infty}\epsilon_n\frac{A(n)\dotsm A(1) A(0)}{\|A(n)\dotsm A(1) A(0)\|}=V\trsp{W}.\]
A vector of initial conditions $U_0$ is called \emph{generic} when $\trsp{W}U_0\neq0$.
\end{proposition}
While in the case of constant coefficients in \cref{lemma:power-method}, the vector~$W$ is an eigenvector of~$\trsp{A}$ and thus can be taken with algebraic coefficients, the situation in the more general situation is more involved and more difficult to control.
\begin{example}
Ap\'ery's recurrence is
\[(n + 2)^3u_{n + 2} =(2n + 3)(17n^2 + 51n + 39)u_{n + 1} -(n+1)^3u_n.\]
Up to constant factors, the vectors turn out to be
\[V=(1,(3+2\sqrt2)^2),\qquad W=(1,-5+6/\zeta(3)).\]
As $\zeta(3)\not\in\mathbb Q$, it follows from Friedland's theorem that any choice of initial conditions in~$\mathbb Q^2$ is generic in this example.
\end{example}
\begin{corollary}\label{coro:Friedland}
Under the hypotheses of \cref{thm:Friedland}, let $K$
be a cone contracted by~$A$, 
let $U_0$ be
generic in the sense of \cref{thm:Friedland} and let the sequence $
(U_n)_n$ be defined by
the recurrence~$U_
{n+1}=A(n)U_n$ for $n\ge0$. Then, there exists~$m\ge0$, such that for all~$n\ge
m$, either~$U_n\in K$ or $U_n\in -K$.
\end{corollary}
\begin{proof}By
\cref{thm:Vandergraft}, one of~$V,-V$ is in
$K^\circ$. With $U_0$ generic, \cref{thm:Friedland} implies that
there exists a sequence $(\epsilon_n)$ in $\{-1,1\}^\mathbb
N$ and $c\neq0$ such that
\[
 \lim_{n\rightarrow\infty}\epsilon_n\frac{U_n }{\|U_n\|}=cv.\]
As this limit is in the interior of either~$K$ or~$-K$, the conclusion
follows.
\end{proof}

\subsection{Proof of \cref{theorem:main} (termination of the algorithm)}
The hypotheses of \cref{theorem:main} are a special case of those
of Friedland's theorem
(\cref{thm:Friedland}), so that for generic initial conditions, the direction of $U_n$
tends towards that of the eigenvector corresponding to the unique dominant eigenvalue $\lambda$ of $A$, which, being unique, is real. 

By \cref{pro:ConeExistence},
there exists a positive cone $K \subset \mathbb{R}^d_{\ge0}$ that is contracted by $A$. This cone is computed in Step~1 using  Vandergraft's construction. In Step 2, by \cref{lem:Contraction}, the convergence of $A(n)$ to $A$ guarantees the existence of an index $m$ such that $A(n)K \subset K$ for all $n \geq m$.
Finally, the termination of Steps~3 and~4 for generic initial conditions is a consequence of \cref{coro:Friedland}.

\subsection{Proof of \cref{thm:decidability}}
 \begin{algorithm}[t]
\caption{Positivity algorithm for scalar recurrences}
\label{algorithm:PP-scalar}
\KwIn{ A recurrence of the form \eqref{rec} with $p_i\in\mathbb Q[n]$ for all~$i$ and $0\not\in p_0p_d(\mathbb N)$;
 initial conditions~$(u_0,\dots,u_{d-1})\in\mathbb Q^d$.}
\emph{It is assumed that the recurrence is of Poincaré type and has a unique simple dominant eigenvalue~$\lambda$.}\\
\KwOut{One of Positive, Non-Positive}

0.\hspace*{0.2cm}{\bf if }{$\lambda>0$ }{\bf then }\\
1.\hspace{0.5cm}Use \cref{algorithm:PP} with $A(n)$ from \cref{eq:companion} and $U_0=(u_0,\dots,u_{d-1})$;\\
2.\hspace*{0.2 cm}{\bf else} \hspace*{0.3 cm}{\bf for }$i=0,\dots$ {\bf do }{\bf if }{$u_i\not\ge0$ }{\bf then }{\bf return} Non-positive\\
\end{algorithm}
In the case of scalar recurrences like \cref{rec}, the reduction to \cref{algorithm:PP} is made by \cref{algorithm:PP-scalar}.

The recurrence operator \( A(n) \) associated with the recurrence by \cref{eq:companion} is defined for all~$n$ thanks to the hypothesis~$0\not\in p_d(\mathbb N)$. It is assumed to be of Poincaré type with a limit matrix \( A \) that has a single dominant eigenvalue~$\lambda$ which is nonzero since~$A\neq0$. A corresponding eigenvector is~$V_\lambda=(1,\lambda,\lambda^2,\dots,\lambda^{d-1})$, which is positive when~$\lambda>0$.

Applicability of \cref{theorem:main} requires that $A(n)$ be
invertible for all~$n$ and that is a consequence of the hypothesis
that $p_0 (n)$ does not have zeros in~$\mathbb N$. The case when $\lambda>0$ is handled by Step~1, whose correctness follows from \cref{theorem:main}.

If $\lambda<0$, then $-A$ has a positive dominant eigenvalue with the eigenvector $V_{\lambda}$. Thus by \cref{pro:ConeExistence}, it contracts a cone~$K$ all of whose elements have coordinates with alternating signs. Then by \cref{coro:Friedland}, for generic~$U_0$ and sufficiently large $n$, $U_n$  has at least one negative coordinate.
This is detected in Step~0 of \cref{algorithm:PP-scalar} and a corresponding index where~$u_n<0$ is computed in Step~2, which therefore terminates for generic initial conditions.

\section{Constructions of Contracted Cones}\label{section:constructions}
We consider a matrix \( A \) with
characteristic polynomial
\begin{equation}\label{eq:charpoly}
\chi(X) = \prod_{i=1}^k (X - \lambda_i)^{m_i},\quad m_1=1. \end{equation}
It is assumed to have a unique simple dominant eigenvalue
$\lambda_1>0$, i.e., $$\lambda_1>| \lambda_2|\geq |\lambda_3|\geq \dots \geq |\lambda_k|.$$

There are several ways to construct cones contracted by~$A$ that lead to distinct cones with different properties for the algorithm. 
We describe two such constructions. The first one is Vandergraft's original construction. It leads to cones that are not polyhedral in general, which makes operating with them expensive. The second one is a variant that produces polyhedral cones.

\subsection{Vandergraft's Construction}\label{subsec:vandergraft}
We recall Vandergraft's construction of a contracted
cone
satisfying \cref{thm:Vandergraft}, using the eigenstructure of \( A
\).

Fix \( \varepsilon \) such that \( \lambda_1 - |\lambda_2| > 
\varepsilon > 0 \) and find a basis~$(V_{i,j})$ of $\mathbb R^d$ that
satisfies
\begin{equation}\label{Basis}
\begin{aligned}
     A V_{i,1} &= \lambda_i V_{i,1},&\quad& 1\le i\le k \\
     A V_{i,j} &= \lambda_i V_{i,j} + \varepsilon V_{i,j-1}, &\quad& 1\le i\le k,\ 2\le j\le m_i.
\end{aligned}
\end{equation}
In other words, in this basis, the matrix $A$ has a block Jordan form
\[
\begin{bmatrix}
    \lambda_1 & 0 & \cdots & 0 \\
    0 & J_{2,\varepsilon} & \cdots & 0 \\
    \vdots & \vdots & \ddots & \vdots \\
    0 & 0 & \cdots & J_{k,\varepsilon}
\end{bmatrix},\qquad
J_{i,\varepsilon} =
\begin{bmatrix}
    \lambda_i & \varepsilon & 0 & \cdots & 0 \\
    0 & \lambda_i & \varepsilon & \cdots & 0 \\
    0 & 0 & \lambda_i & \ddots & \vdots \\
    \vdots & \vdots & \ddots & \ddots & \varepsilon \\
    0 & 0 & \cdots & 0 & \lambda_i
\end{bmatrix},\quad  2\le i\le k.
\]
The cone \( K(V_{i,j}) \) is then defined as:
\begin{multline}\label{eq:cone-Vandergraft}
K(V_{i,j}) = \Bigl\{ a_{1,1} V_{1,1} + \sum_{i=2}^{k} \sum_{j=1}^{m_i} a_{i,j} V_{i,j}\in\mathbb R^d  \Bigm\vert a_{1,1} \geq |a_{i,j}|\text{ for all $i,j$},\\
a_{i,j} = \bar{a}_{p,q} \text{ if } V_{i,j} = \bar{V}_{p,q}. \Bigr\}
\end{multline}
It is easy to check that it is contracted by \( A \) \cite{Vandergraft1968}.

\subsection{Rescaling}
Step~1 of 
\cref{algorithm:PP} is entered with a matrix~$A$ that has a
positive eigenvector associated to~$\lambda_1$. Taking an arbitrary
such vector as~$V_{11}$ in Vandergraft's construction is not
sufficient to produce a cone in~$\mathbb R_ {\ge0}^d$.
However, this can always be achieved after rescaling
\(V_ {1,1}\).

\begin{figure}
\centering
\includegraphics[height=5cm]{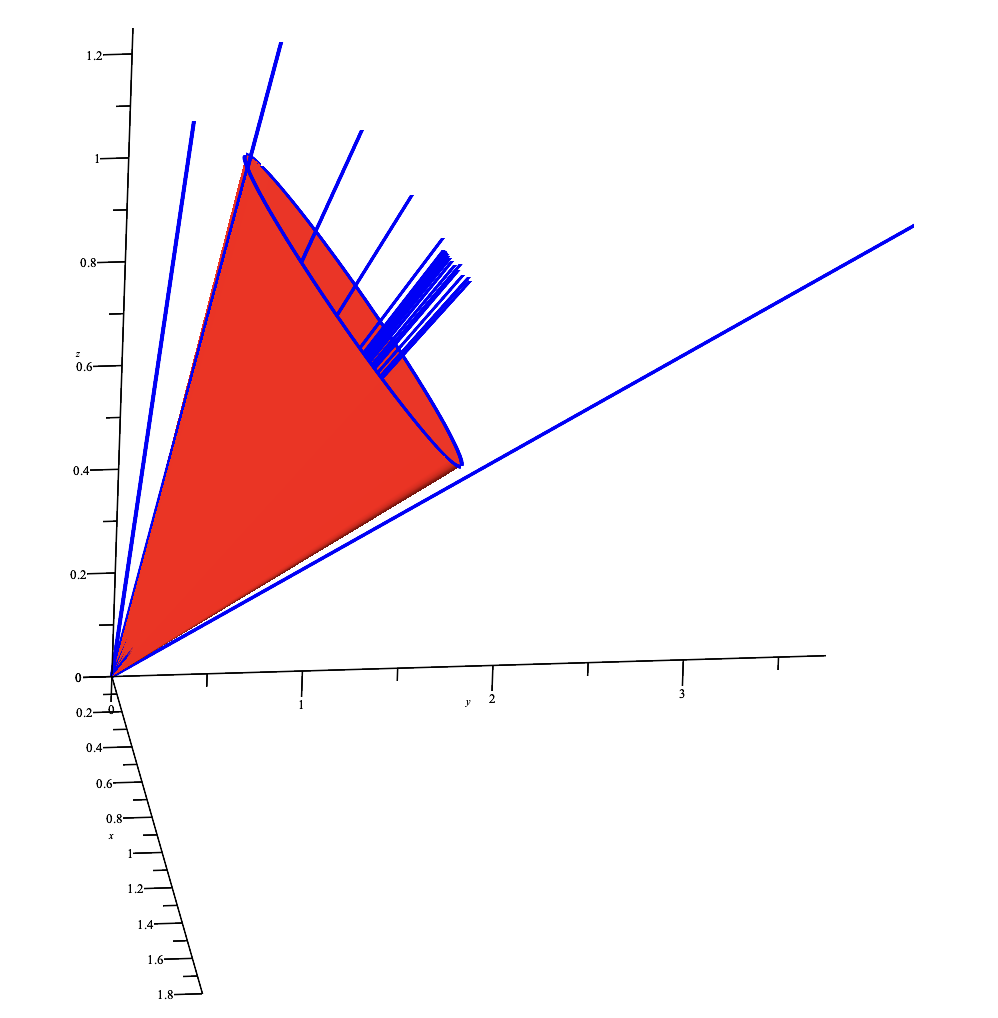}
\caption{Cone from \cref{example:ConeVandergraft} for the recurrence of \cref{example:3}.}
\label{fig:nonpolyhedral}
\end{figure}

\begin{example}\label{example:ConeVandergraft}
Let \( A \) be the limit matrix of the recurrence operator associated with the recurrence in \cref{example:3}. The matrix \( A \) is the companion matrix of the characteristic polynomial
\[
\chi(X) = (X - 1)\left(X - \frac{2+i}4\right)\left(X - \frac{2-i}{4}\right).
\]
Thus the eigenvalues are $\lambda_1=1,\lambda_2=(2+i)/4,\lambda_3=
(2-i)/4$.
Taking the vectors~$V_{i1}=(1,\lambda_i,\lambda_i^2)$
and rescaling~$V_{11}$ by a factor~2 gives a positive cone \( K \) that is the set of all positive scalar multiples of the vectors
\[
\mathbf{v}(a, b) = \left( 2 + 2a, 2 + a - \frac{b}{2}, 2 + \frac{3a}{8} - \frac{b}{2} \right),\qquad a^2 + b^2 \leq 1.
\]
The condition on the parameters \( a \) and \( b \) is not linear,
leading to a non-polyhedral cone depicted in 
\cref{fig:nonpolyhedral}.
\end{example}

\subsection{Polyhedral contracted cones}\label{sec:polyhedral-cones}
As will be apparent in the next section, having a polyhedral cone rather than a general one simplifies greatly the computations in the algorithms. As a special case of a more general result, Tam and Schneider have shown that in the conditions of Vandergraft's theorem (\cref{thm:Vandergraft}), there exists a contracted cone~$K$ that is polyhedral~\cite[Thm.~7.9]{TamSchneider1994}, for which they give an explicit construction~\cite[Proof of Lemma~7.5]{TamSchneider1994}. In the notation of \cref{subsec:vandergraft}, their construction uses the convex hull of the iterates of $A/\lambda_1$ restricted to the space generated by the $V_{i,j}$ where~$(i,j)\neq(1,1)$. As the spectral radius of this restriction is smaller than~1, this convex hull is generated by a finite set of vectors, to which $V_{1,1}$ is then added. That construction has been refined by Valcher and Farina~\cite{ValcherFarina2000}, by exploiting more closely the Jordan form of~$A$.

\begin{example}
For the matrix of \cref{example:ConeVandergraft}, the polyhedral cone constructed by the method of Valcher and Farina is given in \cref{fig:valcher-farina}. While polyhedral, in this and similar examples that we tried, this cone is not very well suited to our method, as it tends to have a large number of faces.
\begin{figure}[h]
    \centering
    \includegraphics[height=4cm]{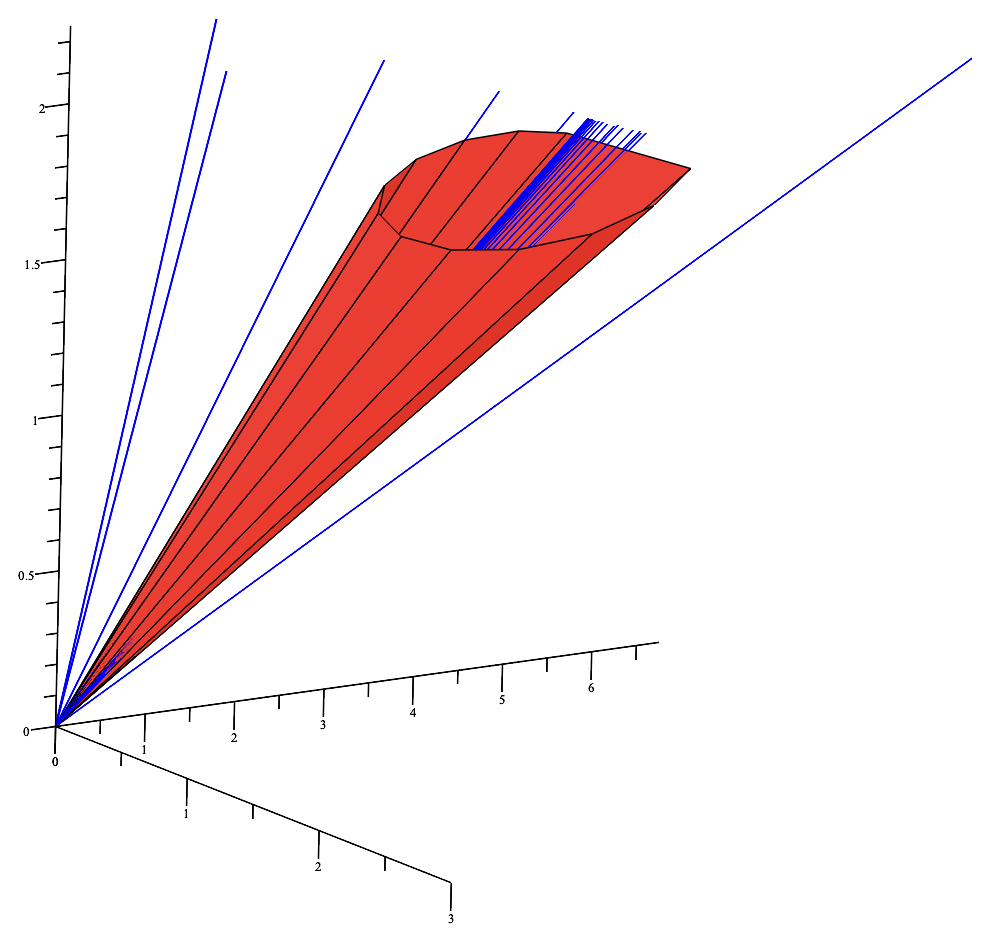}
    \caption{The first values of $U_n=(u_n,u_{n+1},u_{n+2})$ (in blue) in \cref{example:3}, together with the corresponding  polyhedral cone (in red) with its extremal vectors (in black)}\label{fig:valcher-farina}
\end{figure}
\end{example} 

We now propose a modification of Vandergraft's construction that works well in practice by producing polyhedral cones with a small number of faces. 

We start with the simplest case where cones for the 1-norm can be used.
We write $\|z\|_1:=|\Re z|+|\Im z|$ for the 1-norm on complex numbers considered as a real vector space. It satisfies the triangular inequality and is a sub-multiplicative norm ($\|yz\|_1\le \|y\|_1\|z\|_1$ for all complex numbers $y,z$).
\begin{proposition}\label{prop:regularpoly2}
Let \( A \) be a real matrix with a simple dominant eigenvalue \( \lambda_1 > 0 \). Assume that for $i\neq1$, the eigenvalues $\lambda_i$ of \( A \) satisfy the inequality $\|\lambda_i\|_1<\lambda_1$.
Let $V_{i,j}$ be a basis of $\mathbb R^d$ satisfying \cref{Basis},
with \(0<\varepsilon< \lambda_1 -\|\lambda_i\|_1 \) for all $i\neq1$. Then, the polyhedral cone
\[
L(V_{i,j}) = \Bigl\{ a_{1,1} V_{1,1} + \sum_{i=2}^{k} \sum_{j=1}^{m_i} a_{i,j} V_{i,j}\in\mathbb R^d \Bigm\vert a_{1,1} \geq \|a_{i,j}\|_1,\ a_{i,j} = \bar{a}_{p,q} \text{ if } V_{i,j} = \bar{V}_{p,q}\Bigr\}
\]
is contracted by \( A \).
\end{proposition}

\begin{proof}
The cone $L(V_{i,j})$ is a polyhedral cone since it is described by a
finite number of linear inequalities and equalities, which makes
it the intersection of finitely many half spaces in $\mathbb{R}^d$.

Let \( W = \sum a_{i,j} V_{i,j} \) be a vector in the cone \( L(V_{i,j}) \). The vector \( AW \) is
\[
AW = \sum \alpha_{i,j} V_{i,j},
\]
where, by \cref{Basis}, the coefficients are 
\[
\begin{aligned}
\alpha_{i,j} &= a_{i,j} \lambda_i + \varepsilon a_{i,j+1}, \quad \text{for } j < m_i, \\
\alpha_{i,m_i} &= a_{i,m_i} \lambda_i.
\end{aligned}
\]
The property on conjugate coefficients of conjugate vectors follows by taking the conjugates in each case. 
For $j<m_i$, the properties of the norm give
\[\|\alpha_{i,j}\|_1\le \|a_{i,j}\|_1 \|\lambda_i\|_1 + \varepsilon \|a_{i,j+1}\|_1
< a_{1,1}(\lambda_1-\varepsilon)+\varepsilon a_{1,1}=\lambda_1a_{1,1}=\alpha_{1,1},\]
where the last equality comes from the fact that $\lambda_1$ is simple. The case $j=m_i$ gives the same derivation with $\varepsilon$ replaced by~0 and completes the proof that the cone $L(V_{i,j})$ is contracted by the matrix~$A$.
\end{proof}
\begin{example}
\begin{figure}
    \centering
    \includegraphics[width=0.25\linewidth]{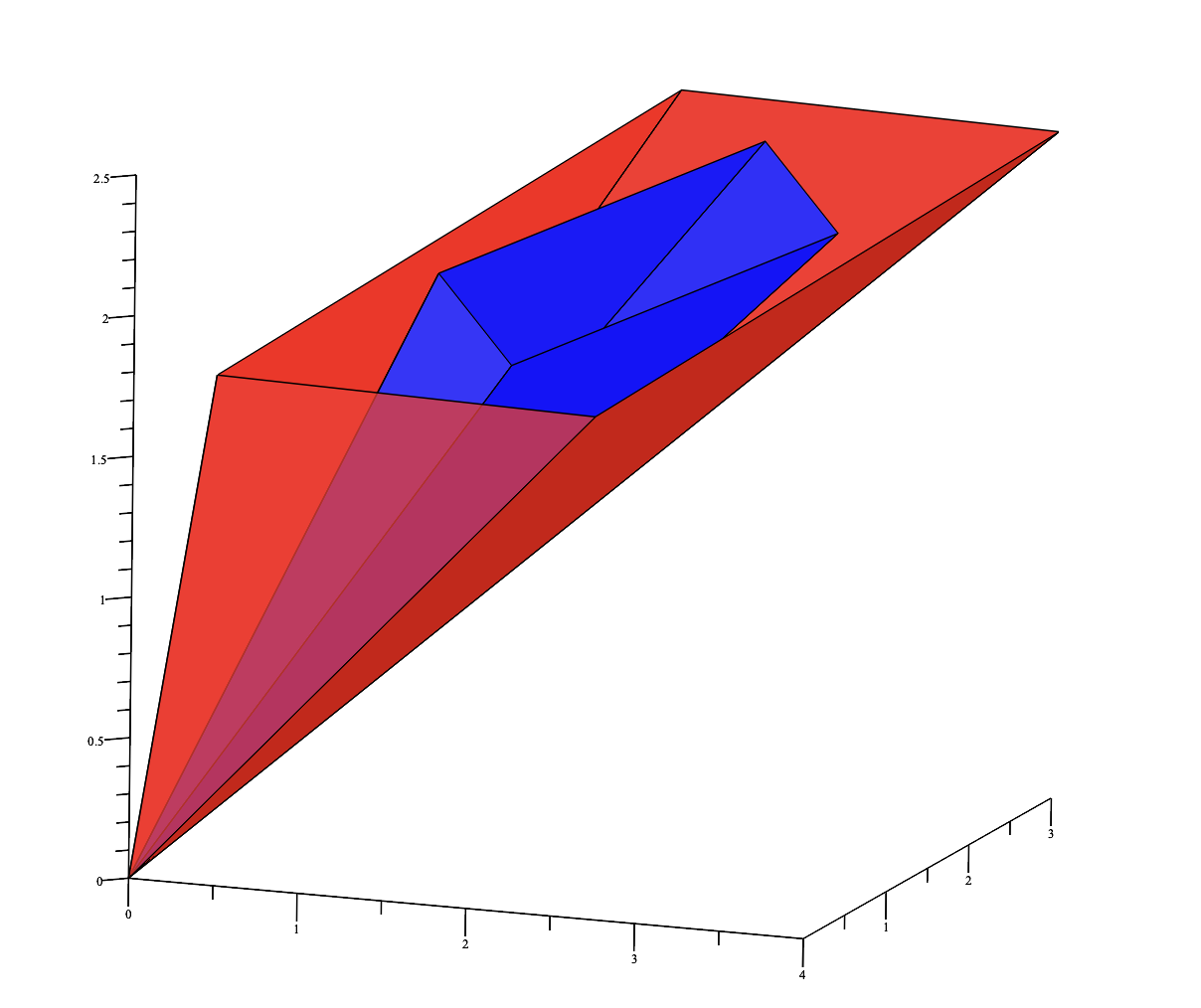}
    \caption{Polyhedral cone (red) with its image by the matrix $A$ (blue)}\label{fig:losange}
\end{figure}
The regular polyhedron given by \cref{prop:regularpoly} for the matrix
\(A \) in \cref{example:ConeVandergraft} is given in \cref{fig:losange}.
\end{example}

For the general case, let $s\in\mathbb N_{>0}$, let $\omega_{s}=\exp(i\pi/s)$ and consider the polygon
\[\mathcal P_s:=\Bigl\{\sum_{i=0}^{2s-1}a_i\omega_{s}^i\Bigm\vert a_i\ge0,a_0+\dots+a_{2s-1}=1\Bigr\},\]
with vertices at the $2s$th roots of~1. The Minkowski functional associated to it is
\begin{equation}\label{eq:normps}
\|z\|_{\mathcal P_s}:=\inf\{r>0\mid z\in r\mathcal P_s\}.
\end{equation}
For~$s=1$ this is the absolute value in~$\mathbb R$. For $s\ge2$, it is a norm on~$\mathbb C$ seen as a real vector space. It is sub-multiplicative: the product of two elements of~$\mathcal P_s$ is a convex linear combination of the $\omega_{s}^i$. The case $s=2$ corresponds to the 1-norm of the previous result. The general version of that result is obtained by merely changing the norm.
\begin{proposition}\label{prop:regularpoly}
Let \( A \) be a real matrix with a simple dominant eigenvalue \( \lambda_1 > 0 \). Then for $i\neq1$, there exists $s_i$ such that  the eigenvalue $\lambda_i$ of \( A \) satisfies the inequality $\|\lambda_i\|_{\mathcal P_{s_i}}<\lambda_1$.
Let $V_{i,j}$ be a basis of $\mathbb R^d$ satisfying \cref{Basis},
with for each $i\neq1$, \( 0<\varepsilon_i< \lambda_1 -\|\lambda_i\|_{\mathcal P_{s_i}} \). Then, the polyhedral cone
\begin{multline}\label{eq:cone-polyhedral}
L(V_{i,j}) = \Bigl\{ a_{1,1} V_{1,1} + \sum_{i=2}^{k} \sum_{j=1}^{m_i} a_{i,j} V_{i,j} \Bigm\vert a_{1,1} \geq \|a_{i,j}\|_{\mathcal P_{s_i}}\text{ for all $i,j$},\\
 a_{i,j} = \bar{a}_{p,q} \text{ if } V_{i,j} = \bar{V}_{p,q}, \Bigr\}
\end{multline}
is contracted by \( A \). 
\end{proposition}
\begin{proof}
The existence of~$s_i$ comes from the fact that $\mathcal P_s$ tends to the unit circle as $s\rightarrow\infty$.
The proof that the cone is contracted is exactly the same as before, up to the change of sub-multiplicative norm.
\end{proof}
In practice, we have observed that the case $s=2$ is often sufficient. In particular, for the large recurrences of \cref{subsec:GRZ} coming from the Gillis-Reznick-Zeilberger sequence, only that case is needed, leading to fast positivity proofs. 

\subsection{Approximate Cones}\label{sec:approx-cone}
The cones $K(V_{i,j})$ and $L(V_{i,j})$ constructed in 
\cref{eq:cone-Vandergraft,eq:cone-polyhedral} are expressed in terms
of the eigenvectors of the matrix~$A$, whose coordinates can be
taken in a finite extension of~$\mathbb Q$. Computations can be
accelerated by selecting a nearby cone whose vertices have simpler coordinates. That this is possible follows from the following.
\begin{proposition}
Let \( A \) be a real matrix, and let \( V_{i,j} \) denote a basis
that satisfies \cref{Basis}. Assume that the cone \( K = K(V_{i,j}) \)
defined
by \eqref{eq:cone-Vandergraft} is contracted by \( A \). Then, there exists \( \eta > 0 \) such that if \( \tilde{V}_{i,j} \) is another basis satisfying \( \|\tilde{V}_{i,j} - V_{i,j}\| < \eta \) for all \( i,j \), the cone \(\tilde K= K(\tilde{V}_{i,j}) \) is also contracted by \( A \). Moreover, the same holds for the cone \( L(V_{i,j}) \) and its approximation \( L(\tilde{V}_{i,j}) \).
\end{proposition}
\begin{proof}
By linearity of \( A \), it is sufficient to prove the result for the intersection $S_0$ of  \(\tilde K \) with the ball $B(0,1)$.

Consider the linear automorphism of $\mathbb R^d$ defined by
\[\phi:W=\sum{\alpha_{i,j}V_{i,j}}\mapsto\tilde W=\sum{\alpha_{i,j}\tilde V_{i,j}}\]
and let $\psi$ denote its inverse. 
The definitions of $K$ and $\tilde K$ imply that $\phi(K)=\tilde K$, $\psi(\tilde K)=K$, $\phi(K^\circ)=\tilde K^\circ$ and $\psi(\tilde K^\circ)=K^\circ$. 
Also, for any~$W$,
\begin{equation}\label{eq:normphi}
\|\phi(W)-W\|\le \sum|\alpha_{i,j}|\|\tilde V_{i,j}-V_{i,j}\|\le \eta\sum|\alpha_{i,j}|,
\end{equation}
and the same bound holds for $\|\psi(\tilde W)-\tilde W\|$.

For all $\tilde W\in S_0$, 
\[\|A\tilde W-A\psi(\tilde W)\|\le  \|A\|\eta C,\]
where $C$ is a finite upper bound on the sum of absolute values of the coordinates of the elements of the compact set $S_0$.
As $\psi(\tilde K)=K$ and $A$ contracts~$K$, the set $S_1:=A\psi(S_0)$ is a compact subset of~$K^\circ$. Thus, there exists $\epsilon>0$ such that $B(x,\epsilon)\subset K^\circ$ for all $x\in S_1$. It follows that if $\eta$ is such that $\|A\|\eta C<\epsilon$, then $A\tilde W\in K^\circ$.
Then, \cref{eq:normphi} implies
\[\|A\tilde W-\phi(A\tilde W)\|<C'\eta,\]
where $C'$ is a finite upper bound on the sum of absolute values of the coordinates of the elements of the compact set $AS_0$.
As $\phi(K^\circ)=\tilde K^\circ$, the set $S_2:=\phi(S_1)$ is a compact subset of $K^\circ$. Thus, there exists $\epsilon'>0$ such that 
$B(x,\epsilon')\subset\tilde K^\circ$ for all $x\in S_2$. It follows that if $\eta$ is such that $C'\eta<\epsilon'$, then $A\tilde W\in\tilde K^\circ$.

In summary, if $\eta$ is such that $\|A\|\eta C<\epsilon$ and $C'\eta<\epsilon'$, then $A (\tilde K\cap B(0,1))\subset \tilde K^\circ$, showing the desired contraction.

The same proof applies to $L$ instead of~$K$.
\end{proof}
Algorithmically, a rational approximation \( \tilde{V}_{i,j} \) of
the basis vectors \( V_{i,j} \) is computed up to a specified
precision. This precision is increased until the resulting cone
\(K (\tilde{V}_ {i,j}) \) (or $L(\tilde V_{i,j})$) generated by this approximation is  contracted by \( A \).

\section{Computations with Cones}
\subsection{Cone Membership}

The cones \( K( V_{i,j}) \) and \( L(V_{i,j}) \) are defined by the set of linear combinations of the vectors \( V_{i,j} \), subject to specific constraints based on the properties of the eigenvalues of the matrix \( A \). The use of both cones relies on checking whether a given vector can be expressed as a valid linear combination of the basis vectors \(V_{i,j} \) under these constraints. 

If \( T \) is the matrix whose columns are the vectors \( V_{i,j} \), the coefficients $a_{i,j}$ of a vector~$W\in\mathbb R^d$ in the basis $(V_{i,j})$ are given by $T^{-1}W$. Then checking membership in $K(V_{i,j})$ or $L(V_{i,j})$ amounts to checking the inequalities
\begin{equation}\label{eq:inside-cone}
\begin{cases}
|a_{i,j}|\le a_{1,1}&\text{ if }a_{i,j}\in\mathbb R,\\
a_{i,j}a_{p,q}\le a_{1,1}^2&\text{ if }V_{i,j}=\bar V_{p,q}\text{ for $K(V_{i,j})$},\\
\|a_{i,j}\|_{\mathcal P_{s_i}}\le a_{1,1}&\text{ for $L(V_{i,j})$}.
\end{cases}
\end{equation}
For $s>1$, the norm~$\|\cdot\|_{\mathcal P_s}$ from \cref{eq:normps} can be computed by
\begin{equation}\label{eq-norm}
\|z\|_{\mathcal P_s}=\Re\left(\frac{1+\omega_s^{-1}}{1+\cos\!\left(\frac\pi{s}\right)}\omega_s^{-m}z\right),\quad\text{with $m$ s.t. }\frac{m\pi}s\le\arg z\le\frac{(m+1)\pi}s.
\end{equation}
Similar inequalities are used for the approximate cones $K(\tilde V_
{i,j})$ and $L (\tilde{V}_{i,j})$.

\subsection{Extremal Vectors}
The maps $A$ and $A(n)$ are linear. Since the vectors in a proper cone \( K \) are nonnegative combinations of its extremal vectors, in order to check that the cone is contracted by one of these maps, it is sufficient to test that the images of the extremal vectors of $K$ belong to the cone.

In the four cases of interest here, Vandergraft's cone $K(V_{i,j})$ of \cref{eq:cone-Vandergraft}, the polyhedral cone $L(V_{i,j})$ of \cref{eq:cone-polyhedral} or their approximation of \cref{sec:approx-cone}, the vectors $V_{i,j}$ are given, together with the distinct complex numbers $\lambda_i$ and they satisfy the property that if $\lambda_p=\bar\lambda_i$ then $V_{p,j}=\bar V_{i,j}$ for $j=1,\dots,m_i$. Without loss of generality, the $\lambda_i$ are numbered so that $\lambda_1,\dots,\lambda_{v_r}$ are real,  $\lambda_{v_r+1},\dots,\lambda_{v_c}$ have positive imaginary part and $\lambda_{v_c+i}=\bar\lambda_{v_r+i}$ for $i=1,\dots,v_c$.

Up to normalization so that their coordinate on the vector~$V_{1,1}$ is~1, the extremal vectors of the cone \( K(V_{i,j}) \), denoted by \( \mathrm{Ext}(K) \), are given by
\[\label{Ext}
\mathrm{Ext}(K) = \Bigl\{ V_{1,1} + \sum_{i=2}^{v_r} \sum_{j=1}^{m_i} \epsilon_{i,j} V_{i,j} + 2
\sum_{i=v_r+1}^{v_c} \sum_{j=1}^{m_i} \Re(\alpha_{i,j}V_{i,j}) \Bigm\vert\epsilon_{i,j}\in\{\pm1\},  |\alpha_{i,j}|^2 = 1 \Bigr\}.
\]
\begin{multline*}
\mathrm{Ext}(L) = \Bigl\{ V_{1,1} + \sum_{i=2}^{v_r} \sum_{j=1}^{m_i} \epsilon_{i,j} V_{i,j} +2 \sum_{i=v_r+1}^{v_c} \sum_{j=1}^{m_i}\Re( \alpha_{i,j} V_{i,j})\Bigm\vert \\
\epsilon_{i,j}\in\{\pm1\}, \alpha_{i,j} \in \{1,\omega_{s_i},\dots,\omega_{s_i}^{2s_i-1}\} \Bigr\}.
\end{multline*}

\subsection{Cone Contraction}
Let $W$ be an extremal vector in the set $\operatorname{Ext}(K)$ or $\operatorname{Ext}(L)$.
In the basis~$(V_{i,j})$, its coordinates are
\[
\begin{cases}
1&\text{if $i=1$;}\\
\epsilon_{i,j}\in\{\pm1\}&\text{if $2\le i\le v_r$;}\\
\alpha_{i,j}&\text{for $v_{r}+1\le i\le v_c$;}\\
\bar \alpha_{i-v_c+v_r,j}&\text{for $v_{c}+1\le i\le 2v_c-v_r$},
\end{cases}
\]
with the constraints $|\alpha_{i,j}|^2=1$ for the cone~$K$ and $\alpha_{i,j}\in\{1,\dots,\omega_{s_i}^{2s_i-1}\}$ for the cone~$L$.

Let $T$ be the matrix whose columns contain the coordinates of the vectors of the basis~$(V_{i,j})$. The coordinates of the image by~$A$ of the extremal vector~$W$ on this basis are the entries of the vector
$T^{-1}ATW$.
Denoting these coordinates by~$a_{i,j}$, checking that the image belongs to the cone reduces to checking the inequalities~\eqref{eq:inside-cone}.

Let $\mathbb K$ be a field containing~$\mathbb Q$, the entries of $A$ and the real and imaginary parts of the coordinates of the vectors $V_{i,j}$. In particular, $\mathbb Q$ can be taken for~$\mathbb K$ when using the approximate cones of \cref{sec:approx-cone}. Let $\ba,\bb$ be the real and imaginary parts of the coordinates~$\alpha_{i,j}$ of $W$.

In the case of the cone~$K(V_{i,j})$, the inequalities \cref{eq:inside-cone} reduce to a finite set of inequalities between polynomials in~$\mathbb K[\ba,\bb]$, subject to $|\alpha_{i,j}|^2=1$, that form a set of nonlinear constraints on~$\ba,\bb$.
In the case of the cone~$L(V_{i,j})$, the situation is simpler in that no parameter is involved and the system gives a finite set of inequalities, with coefficients in an algebraic extension to accommodate the sines and cosines of \cref{eq-norm} when one of the $s_i$ is larger than~2.

\subsection{Stability Index}
Given a matrix $A(n)$ and a cone $K$ contracted by its limit matrix
$A$,  the second step of \cref{algorithm:PP}  computes a
stability index, i.e., an integer $m\in\mathbb{N}$ such that $A(n)K\subset K$ for all $n\geq
m$. In the notation of the previous section, the system of interest is given by the inequalities~\eqref{eq:inside-cone} applied to the coordinates~$a_{i,j}$ of the vector
$T^{-1}A(n)TW$, where $W$ runs over the extremal vectors of the cone.

 By assumption, $A(n)$ tends to~$A$ for which the cone is contracted, so that these inequalities and constraints are satisfied simultaneously for large enough~$n$. First, one can compute an integer~$m_0$ such that the denominator of~$A(n)$ has constant sign for $n\ge m_0$. Assuming $n\ge m_0$, multiplying the inequalities by this denominator gives a new set of inequalities that are \emph{polynomial} in~$n$ and can be written in the form~$p_i(n)\ge0$ for $p_i\in\mathbb K[\ba,\bb]$.
 
 This construction gives a set of polynomial inequalities and
 equalities that define semi-algebraic sets whose projection on the
 $n$-axis contain an interval of the form~$[m_1,\infty)$. Such a~$m_1$
 could be found by quantifier elimination. However, the optimal~$m_1$
 is not needed and computing it might be costly. Instead, we construct
 polynomials~$\tilde p_i(n)\in\mathbb Q[n]$ by computing lower bounds
 on each of the coefficients of the $p_i$ by interval analysis. As
 $\tilde p_i(n)\ge p_i(n)$ for all~$n\ge0$, it is then sufficient to
 find an upper bound on the largest real roots of these polynomials,
 which can be computed efficiently. A detailed example using this technique is given in \cref{sec:example}.

\section{Improvements for Scalar Recurrences}\label{sec:linrec}
When the matrix~$A(n)$ is a companion matrix as in 
\cref{eq:companion}, extra structure can be exploited in order to make
the algorithm more efficient.

\subsection{Larger cones}
In Step~1 of the algorithm, it is sufficient to have $K \subset 
\mathbb{R}^{d-1} \times \mathbb R_{\ge0}$, rather than requiring it to
be inside $\mathbb R_{\ge0}^d$. In this case, proving that $U_n = 
(u_n, u_{n+1},
\dots, u_{n+d-1}) \in K$ ensures that $u_{n+d-1} \geq 0$ for all $n$,
which implies the positivity of the entire sequence after checking
initial conditions. The cone being larger, the sequence enters it for smaller indices and fewer initial conditions have to be checked.

\subsection{Explicit basis}
If the roots~$\lambda_1,\dots,\lambda_k$ of the characteristic
polynomial are known, with $m_i$ the multiplicity of $\lambda_i$, then the
following vectors satisfy the condition of \cref{Basis}:
\begin{equation}\label{eq:basis}\begin{aligned}
    V_{i,1} &= (1, \lambda_i, \dots, \lambda_i^{d-1}),&& 1\le i\le k \\
    V_{i,j} &= \varepsilon^{j-1} \left( \binom{\ell}{j-1} \lambda_i^
    {\ell - j + 1}, \ \ell = 0, \dots, d-1 \right), && 1\le i\le k,\
    2\le j\le m_i.
\end{aligned}\end{equation}
In the special situation when $\lambda_k=0$, this is understood in
the limit of this expression as $\lambda_k\rightarrow0$, namely
\[V_{k,j}=\varepsilon^{j-1}(0,\dots,0,1,0,\dots,0),\]
with~1 in the $j$th position.

 \section{Detailed Example}\label{sec:example}
We illustrate the steps of the algorithm on the sequence \( (u_n^{(4)})\) from \cref{eq:GRZ} with both the cone from Vandergraft’s construction and its modified polyhedral analogue. The sequence is a solution of the recurrence relation:
\begin{multline*}
    \big(10616832n^6 + 138018816n^5 + 701374464n^4 + 1765380096n^3 +\\  2308829184n^2 + 1500622848n + 383201280\big)u_n \\
    + \big(1769472n^6 + 25657344n^5 + 150073344n^4 + 453150720n^3 +\\ 746896896n^2 + 640811520n + 224985600\big)u_{n+1} \\
    + \big(110592n^6 + 1769472n^5 + 11582208n^4 + 39696768n^3 +\\ 75175488n^2 + 74657088n + 30421440\big)u_{n+2} \\
    + \big(-5120n^6 - 89600n^5 - 647744n^4 - 2473952n^3\\ - 5258744n^2 - 5889032n - 2708160\big)u_{n+3} \\
    + \big(32n^6 + 608n^5 + 4738n^4 + 19353n^3 +\\ 43628n^2 + 51376n + 24640\big)u_{n+4} = 0.
\end{multline*}
with initial conditions $u_0 = 1$, $u_1 = 0$, $u_2 = 216$, $u_3 = 18816$.
\subsection*{Step 0: Eigenvalues}  
The characteristic polynomial of the recurrence is
\begin{gather*}
\chi(X) = X^4 - 160X^3 + 3456X^2 + 55296X + 331776 = \prod_{i=1}^4 (X - \lambda_{i}),\\
\text{with}\quad
\lambda_1 \approx 129.9898, \quad \lambda_2 \approx 42.0400, \quad \lambda_3 \approx -6.0149 + 4.9530i, \quad \lambda_4 = \bar{\lambda}_3.
\end{gather*}
The assumption of having one simple dominant eigenvalue \( \lambda_1 > 0 \) is satisfied. The limit $A$ of the recurrence operator \( A(n) \) is the companion matrix of \( \chi \).

\subsection*{Step 1: Positive Contracted Cone}  
The eigenvalues of this recurrence are all non-rational algebraic numbers and give a basis $V_{i,j}$ by \cref{eq:basis} with algebraic coordinates.
Following \cref{sec:approx-cone}, an approximate rational basis \( \tilde{V}_{i,j} \) is computed and used to construct the cone \( K(\tilde{V}_{i,j}) \), which is an approximation of the cone \( K(V_{i,j}) \) up to a certain precision.  Next, the vector $\tilde{V}_1$ is multiplied by a positive real number to make the cone $K(\tilde{V}_{i,j})$ positive.
Finally, the approximate cone is checked to be contracted by \( A \). If not, the precision is increased.

\subsubsection*{Approximate Basis}
With 3 digits of precision, we get the basis $(\tilde{V}_{1},\tilde V_{2},\tilde V_{3},\bar{\tilde V}_{3})$ that forms the columns of the matrix
\begin{equation}\label{eq:matrixT}T=\renewcommand{\arraystretch}{1.3} 
\begin{pmatrix}
1 & 1 & 1 & 1 \\
130 & \frac{421}{10} & -6 + \frac{54}{11}i & -6 - \frac{54}{11}i \\
16800 & 1770 & \frac{58}{5} - \frac{298}{5}i & \frac{58}{5} + \frac{298}{5}i \\
2180000 & 74800 & 225 + 416i & 225 - 416i
\end{pmatrix}. \end{equation}
Both \( K( \tilde{V}_1,\tilde V_2,\tilde V_3,\bar{ \tilde{V}}_3) \) and \( L( \tilde{V}_1,\tilde V_2,\tilde V_3,\bar{ \tilde{V}}_3) \) are cones made of vectors of the form \[T\cdot\trsp{(\alpha_1,\alpha_2,a+ib,a-ib)},\] with distinct constraints:
\begin{alignat*}{3}
   \mathcal C_K &:= \{\alpha_1 \geq |\alpha_2|, \quad \alpha_1^2 \geq a^2 + b^2\}\qquad &&\text{for $K( \tilde{V}_1,\tilde V_2,\tilde V_3,\bar{ \tilde{V}}_3)$},\\
   \mathcal C_L &:= \{\alpha_1 \geq |\alpha_2|, \quad \alpha_1 \geq |a| + |b|\}\qquad&&\text{for $L( \tilde{V}_1,\tilde V_2,\tilde V_3,\bar{ \tilde{V}}_3)$}.
   \end{alignat*} 
\subsubsection*{Rescaling}
We compute  $\beta > 0$ such that the cone $K(\beta \tilde{V}_1, \tilde{V}_2, \tilde{V}_3, \bar{\tilde{V}}_3)$ lies inside $\mathbb{R}^3 \times\mathbb R_{\ge0}$. For this, it is sufficient to find a rational $\beta>0$ such that 
\[2180000\beta - 74800 + 450a - 832b>0, \text{~for all } a,b\in\mathbb{R},~ a^2+b^2\leq 1.\]
The smaller~$\beta$, the larger the cone, which results in a smaller stability index. Here, taking $\beta=1/25$ is sufficient, leading to
\[K_0=T\cdot\trsp{\left(\frac1{25}\alpha_1,\alpha_2,a+ib,a-ib\right)}\quad\text{with}\quad \mathcal C_K.\]
Similarly, the cone $L_0=L(\tilde{V}_1/25,\tilde V_2,\tilde V_3,\bar{ \tilde{V}}_3)$ with $\mathcal C_L$ is in $\mathbb{R}^3 \times\mathbb R_{\ge0}$.
\subsubsection*{Test Contraction}
To verify whether the cone \( K_0 \) is indeed contracted by \( A \), it is sufficient to test that its extremal vectors are mapped inside of it. Up to normalization, the set of extremal vectors of \( K_0 \) is
\[
\text{Ext}(K_0) = \left\{\beta \tilde{V}_1 \pm \tilde{V}_2 + 2 \Re((a + ib)\tilde{V}_3) \, \middle| \, a^2 + b^2 = 1 \right\}.
\]

For instance, taking \( W_{a,b} = \beta \tilde{V}_1 + \tilde{V}_2 + 2 \Re((a + ib) \tilde{V}_3) \), the coefficients~$(\alpha_1,\alpha_2,\alpha_3,\overline\alpha_3)$ of $AW$ in the basis from \cref{eq:matrixT} are the entries of  \( T^{-1} A W \).
The corresponding vector belongs to~$K_0$ if and only if the following three polynomials are positive under the constraint $a^2 + b^2 = 1$:
\begin{align*}
\alpha_1 - \alpha_2 &= \frac{4679376341663687}{1457566040726} - \frac{17478404684}{728783020363} a - \frac{943647073524}{8016613223993} b, \\
\alpha_1 + \alpha_2 &= \frac{33569621452327161}{10202962285082} - \frac{74575443452}{5101481142541} a - \frac{2936449777372}{56116292567951} b, \\
\alpha_1^2 - \alpha_3 \bar{\alpha}_3 &= \frac{60486439965294533189130216454777027}{5725524166494313887186069820} \\
&\quad - \frac{71525124113790875227828970643}{650627746192535668998417025} a \\
&\quad - \frac{3898235264387376435017043223868}{7156905208117892358982587275} b \\
&\quad - \frac{1606389690866415174695257143}{5725524166494313887186069820} a^2 \\
&\quad + \frac{6917145006021706588092982004}{7156905208117892358982587275} ab.
\end{align*}
In view of the small degree of these polynomials, it is not difficult to check their positivity on the circle $a^2+b^2=1$. For efficiency reasons however, it is simpler to check a sufficient condition: replace $b^2$ by $1-a^2$ (as done above) and use interval arithmetic to compute a lower bound on these polynomials on the square $[-1,1]^2$ that contains the circle. For example, for the first polynomial 
\[\alpha_1-\alpha_2\geq 
 \frac{4679376341663687}{1457566040726}- \frac{17478404684}{728783020363} -\frac{943647073524}{8016613223993}>\frac{3617966}{1127}
\]
shows that  \(\alpha_1 - \alpha_2>0\) for all \((a, b) \in [-1,1]^2\), hence also for $(a, b)$ on the unit circle.
If the computed lower bound is not positive on the intervals, we increase the precision of the approximations. This ensures positivity, as the exact cone is contracted by A, and the approximate cone will be sufficiently accurate at higher precision.

This process can be applied to the remaining two polynomials. In the case of \( \alpha_1^2 - \alpha_3 \bar{\alpha}_3 \), as it is not linear in the variables \( a \) and \( b \), interval arithmetic over-approximates the interval containing its value, but still returns a positive interval from which the positivity is deduced.

The same steps show that the remaining extremal vectors of form $\beta \tilde{V}_1 - \tilde{V}_2 + 2 \Re((a+ib)\tilde{V}_3$ are also mapped inside of~$K_0$ by~$A$, confirming that the cone \( K_0 \) is contracted by \( A \).\\

The case of the cone $ L_0$ is similar, except that no variable occur in the extremal vectors, allowing the lower bound to be computed directly. First, 
\[
\text{Ext}(L_0) = \left\{\beta \tilde{V}_1 \pm \epsilon \tilde{V}_2 + 2 \Re((a+ib)\tilde{V}_3) \, \middle| \, \epsilon=\pm 1,~ (a,b) \in \{(\pm 1,0),(0,\pm1)\} \right\}.
\]
Let $W_{L}=\beta \tilde{V}_1+\epsilon \tilde{V}_2+2 \Re((a+ib)\tilde V_3)$ be an extremal vector of $L_0$. With the same notation, $T^{-1}AW_{L}=(\alpha_1\alpha_2,\alpha_3,\bar{\alpha}_3)$. The polynomial system that implies the contraction by $A$ is given by 
$$\alpha_1\pm \alpha_2>0,\quad\alpha_1+ (\pm \Re(\alpha_3) \pm \Im (\alpha_3))>0.$$
All these inequalities are linear. For instance, here is the proof of positivity of one of the polynomials:
\begin{align*}
\small
\alpha_1 &+ \Re(\alpha_3) + \Im(\alpha_3)\\
&= \frac{473760971363066}{3643915101815} + \frac{1168800515801}{457566040726}\epsilon - \frac{1513887146713}{1457566040726}a - \frac{8043173512184}{728783020363}b \\
&\geq \frac{473760971363066}{3643915101815} \textcolor{red}{-} \frac{1168800515801}{457566040726} \textcolor{red}{-} {\frac{1513887146713}{1457566040726}} \textcolor{red}{-} \frac{8043173512184}{728783020363}\\
&=\frac{426838384645861}{3643915101815}>0.
\end{align*}

\subsection*{Step 2: Stability Index}
The computations are similar, but now the vector of coefficients  $(\alpha_1,\alpha_2,\alpha_3,\overline\alpha_3)$ is \( T^{-1} A(n) W \), for each extremal vector~$W$ of $K_0$ (or $L_0$), with a parameter~$n$.
The goal is to find \( n_0 \in \mathbb{N} \) such that the polynomials $\alpha_1-\alpha_2,\alpha_1+\alpha_2,\alpha_1^2-\alpha_3\overline\alpha_3$ in $\mathbb Q[a,b,n]$ are positive for all \( n \geq n_0 \) and \( a, b\) with $a^2+b^2=1$.
These polynomials can be written as
\[P(n,a,b) = \sum_{i=0}^{\deg(P)} P_i(a,b) n^i,\]
where the positivity of the leading coefficient \( P_{\deg(P)}(a,b) \) has already been verified in the previous step (it corresponds to the condition for the limit matrix~$A$).
In order to determine the stability index, we compute a univariate polynomial in \( n \), denoted \( \tilde{P}(n) \), such that
\[\tilde{P}(n) \le P(n,a,b), \qquad n \in \mathbb{N},\quad a^2+b^2=1.\]
This step is performed by computing a lower bound for the coefficients \( P_i(a,b) \) under the constraints $(a,b)\in[-1,1]^2$. For example, the difference \( \alpha_1 - \alpha_2 \) is
{\small
\begin{align*}
\sum_{i=0}^6 P_i(a,b) n^i 
&= \left( \frac{4679376341663687}{1457566040726} - \frac{17478404684}{728783020363} a - \frac{943647073524}{8016613223993} b \right) n^6 \\
&\quad + \left( \frac{79944448001932229}{1457566040726}+ \frac{16416449422588}{728783020363}a - \frac{182507077766172}{8016613223993} b  \right) n^5 \\
&\quad + \left( \frac{9024446755521470599}{23321056651616} + \frac{842882173059821 }{2915132081452}a - \frac{9457184901260669}{32066452895972}b \right) n^4 \\
&\quad + \left( \frac{67295900421104575215}{46642113303232}+ \frac{8296842530571717}{5830264162904} a- \frac{98453187869997173}{64132905791944}b  \right) n^3 \\
&\quad + \left( \frac{34946375928911783293}{11660528325808} + \frac{2446943888987389}{728783020363}a - \frac{62893142299669887}{16033226447986}b\right) n^2 \\
&\quad + \left( \frac{9576203440704527405}{2915132081452}+ \frac{5537344397615665}{1457566040726}a - \frac{39267471669723590}{8016613223993}b \right) n \\
&\quad + \left( \frac{1079862274983488675}{728783020363} +\frac{1215292264210220}{728783020363} a- \frac{19190127116357680}{8016613223993}b \right).
\end{align*}}
For the leading term of this polynomial, \( P_6(a,b) \), a lower bound was found above.  Proceeding similarly to find lower bounds for the other coefficients gives the lower bound
\begin{multline*}
\alpha_1 - \alpha_2\ge
\frac{3617966}{1127} n^6 + \frac{14413088}{263} n^5 + \frac{221783003}{574} n^4\\
+ \frac{18718127}{13} n^3 + \frac{86701307}{29} n^2 + \frac{19657805}{6} n + \frac{13299050}{9},
\end{multline*}
which is positive for all \( n \in \mathbb{N} \), showing that $\alpha_1-\alpha_2$ itself is positive for all \( n \in \mathbb{N} \).

Repeating this process for all the other polynomials confirms that the stability index is \( n_0 = 0 \), even though the nonlinearity of the polynomial \( \alpha_1^2 - \alpha_3 \bar{\alpha}_3 \) with respect to the variables \( a \) and \( b \) induces a loss of accuracy in the interval analysis. 

Using $L_0$ instead of $K_0$ is both more direct and more efficient, since the polynomial system used to determine the stability index is linear in all variables except~$n$.

\subsection*{Step 3: Initial Conditions} The last step checks that the sequence enters the cone. As $n_0=0$, $U_0>0$ and $U_0\in K_0$, this step is immediate and the positivity of the sequence is proved.

\section{Experiments}
\label{section:experiments}
A first implementation of our algorithm has been written in Maple\footnote{The Maple package and a worksheet of examples are available at\\ \url{https://gitlab.inria.fr/alibrahi/positivity-proofs}}. We now report on its behaviour on a selection of tests\footnote{Timings are given on a MacBook Pro equipped with the Apple M2 chip, using Maple2024.}.

\subsection{P-finite Examples}
The positivity of the sequence $(s_n)_n$ in \cref{SZ} is established by our code using the following instruction:
{\color{darkred}\begin{Verbatim}
> Positivity({(2*n^2+8*n+8)*s(n+2)=(81*n^2+243*n+186)*s(n+1)
> -(729*n^2+1458*n+648)*s(n),s(0)=1,s(1)=12},s,n);
\end{Verbatim}
}
\noindent which returns
\[{\color{blue}\text{true},
\begin{pmatrix}
\frac{51}{100} & 1 \\
\frac{1377}{100} & \frac{27}{2}
\end{pmatrix},\begin{pmatrix}
1 & -1 \\
1 & 1
\end{pmatrix},~2
}\]
The first part of the output, ``true'', means that the sequence $(s_n)$ is positive. Next comes a matrix $T$ whose columns are the vectors~$V_{ij}$ used to construct the cone. The next matrix~$M$ defines the cone: the columns of $MT$ are the extremal vectors of the contracted cone $K$, which is polyhedral. Finally, the integer $2$ indicates that $S_{n}=(s_n,s_{n+1})$ belongs to ${K}$ for all $n\geq 2$.

Similar computations with the recurrences
\begin{gather*}
    (n + 1)^3 d_{n + 1} =4(2n + 1) \left( 3n^2 + 3n + 1 \right) d_{n} - 16n^3 d_{n - 1}\\
    (20n + 1) f_{n + 3} = 3 (5n + 1) f_{n + 2} - (13n + 1) f_{n + 1} + (18n + 7) f_{n}\\
    (n + 1)u_{n + 3} = \left(\frac{77}{30}n + 2\right)u_{n + 2} - \left(\frac{13}{6} n - 3\right)u_{n + 1} +\left(\frac{3}{5}n + 2\right)u_{n}
\end{gather*}
with initial conditions $d_0 = 1,d_1 = 4,f_0 = 1,f_1 = 1,f_2 = 3, u_0=1,u_1=15/14,u_2=8/7$ coming from previous works~\cite[Ex.4.2]{StraubZudilin2015}, \cite[Ex.1]{pillwein2015efficient}, \cite{IbrahimSalvy2024}.
For each recurrence, the value \( n_0 \) represents the computed index such that the sequence vectors \( U_n \) satisfy \( U_n \in K \) for all \( n \geq n_0 \), where \( K \) is the constructed cone. The computations yield \( n_0 = 1 \) for \( d_n \), \( n_0 = 12 \) for \( f_n \), and \( n_0 = 1374 \) for \( u_n \). 
 (With our previous method $u_n$ led to~$n_0=3040$~\cite{IbrahimSalvy2024}.)

\subsection{C-finite Sequences}
Nuspl and Pillwein~\cite{nuspl2022comparison} have performed extensive experiments of their method on sequences from the OEIS~\cite{oeis} that satisfy linear recurrences with constant coefficients. In their list, 286 recurrences have a single simple dominant eigenvalue and can be subjected to our method. Their orders range from~1 to~11. In this situation of recurrences with constant coefficients, our algorithm becomes simpler, as the step computed the stability index can be skipped. 
In all the examples, the execution time was smaller than 0.5 seconds, both with the circular cones of Vandergraft's construction and the polyhedral cones of \cref{sec:polyhedral-cones}. The value of \( n_0 \) such that \( U_n \) belongs to the cone for all \( n \geq n_0 \) was always equal to~1. 

For all the~286 recurrences except~14 of them, a polyhedral cone can be computed by \cref{prop:regularpoly2}. One of these 14~exceptional cases is the OEIS sequence A001584, defined by
\[u_{n+8}=2u_{n+5}-u_{n+2}+u_n,\quad u_0=\dots=u_7=1.\]
There, the dominant eigenvalue is $\lambda_1\simeq 1.22$, root of $x^4-x-1$, while the next one in modulus is $\lambda_2\simeq-0.73+0.93i$, root of $x^4-x+1$. For $s<6$ one has $\|\lambda_2\|_{\mathcal P_s}>\lambda_1$ and the first admissible value in \cref{prop:regularpoly} is $s=6$. 
As the values of~$s>2$ used in \cref{prop:regularpoly} have not been implemented yet, Vandergraft's construction is used by our code in this case. 

\subsection{Gillis-Reznick-Zeilberger family}\label{subsec:GRZ}
\begin{table}
    \centering
\begin{tabular}{ccc|cr|cr}
$k$  & max   & max $\log_{10}$ & $n_{K}$ & $t_{K}$ & $n_{L}$ & $t_{L}$ 
\\ 
& deg& coeffs&&&&\\ \hline
 4 & 6 & 10 & 2 &  0.1 & 3 &  0.1 
\\
 5 & 10 & 20 & 2 &  0.3 & 546 &  0.2 
\\
 6 & 15 & 32 & 2 &  0.5 & 3 &  0.2 
\\
 7 & 21 & 52 & 2 &  1.2 & 3 &  0.4 
\\
 8 & 28 & 72 & 5 &  2.0 & 4 &  0.5 
\\
 9 & 36 & 100 & 4 &  6.4 & 4 &  2.4 
\\
 10 & 45 & 129 & 8 &  8.2 & 4 &  1.7 
\\
 11 & 55 & 175 & 30 &  21.5 & 3 &  2.8 
\\
 12 & 66 & 202 & 17 &  23.6 & 4 &  5.3 
\\
 13 & 78 & 270 & 75 &  66.1 & 1 &  7.3
\\
 14 & 91 & 310 & 31 &  64.7 & 1 &  16.5
\\
 15 & 105 & 366 & 39 &  192.5 & 1 &  21.1 
\\
 16 & 120 & 433 & 249 &  220.7 & 1 &  39.8
\\
 17 & 136 & 531 & 59 &  610.8 & 1 &  53.9
\\
 18 & 153 & 569 & 527 &  592.0 & 1 &  78.6
\\
 19 & 171 & 699 & 774 &  2704.8 & 1 &  120.4
\\
 20 & 190 & 739 & 100 &  1171.8 & 1 &  209.1
\\
 21 & 210 & 850 & ---  & ---  & 1 &  446.5
\\
 22 & 231 & 960 & ---  & ---  & 1 &  416.9
\\
 23 & 253 & 1115 & ---  & ---  & 1 &  731.6
\\
 24 & 276 & 1140 & ---  & ---  & 1 &  1827.2 
\end{tabular}
    \caption{Experiments with the Gillis-Reznick-Zeilberger family}
    \label{tab:GRZ}
\end{table}
The family from \cref{eq:GRZ} comes from the diagonal of the rational function
\[A_k(x_1, x_2, \dots, x_k) = \frac{1}{1 - (x_1 + x_2 + \dots + x_k) + k! x_1 x_2 \dots x_k}.\]
For a long time, it was only a conjecture that this sequence was positive for all~$k\ge4$. Thus it became a nice test for positivity-proving algorithms. In~2007, Kauers used a modification of his method with Gerhold~\cite{GerholdKauers2005} to prove the positivity for $k=4,5,6$~\cite{Kauers2007a} with a runtime that was ``no more than a few minutes'' on a computer of that time. In~2019, Pillwein~\cite{Pillwein2019} used a variant of this approach to prove the positivity up to~$k=17$ (but did not report on computation time). A small difficulty is that computing the recurrences is also time-consuming. In our experiment, we used the Maple package \texttt{CreativeTelescoping} that implements a recent reduction-based algorithm~\cite{BrochetSalvy2024} to produce the recurrences for~$k$ up to~24 (the last one has order~24, coefficients of degree~276 with integer coefficients having up to~1140 decimal digits). The results are presented in \cref{tab:GRZ}. The first three columns give data on the recurrence: the index~$k$ of the sequence, which is also the order of the recurrence; the maximum degree of the polynomial coefficients; the maximal number of decimal digits of the coefficients of these polynomials. The next two columns give the results obtained with Vandergraft's cone~$K$ from \cref{subsec:vandergraft}: first the stability index and next the computation time. The last four values have not been computed with this method. The last two columns give the same results when using the polyhedral cone~$L$ of \cref{sec:polyhedral-cones}. For both cases, the computation time is mostly spent in the computation of the stability index. In order to save some time in the case of the cone~$K$, this index is bounded using interval analysis as shown in the detailed example of \cref{sec:example}; this explains why the values of $n_K$  are most often significantly larger than the corresponding~$n_L$ computed with the polyhedral cones.
\section{Conclusion}
The method presented here proves positivity for a large class of linear recurrences. The class can be extended in several directions:
\begin{itemize}
    \item[--] As mentioned before, the conditions that the first and last coefficients of the recurrence do not vanish on~$\mathbb N$ are not really constraints: it is sufficient to check separately finitely many initial terms and shift the recurrence;
    \item[--] Similarly, being of Poincaré type is not a severe restriction, as a positivity-preserving transformation reduces to that situation;
    \item[--] All coefficients of recurrences and initial conditions considered here are rational numbers; the algorithm applies with coefficients in a real number field, it is only necessary to be able to decide equalities and inequalities;
    \item[--] The main constraint is the uniqueness of the dominant eigenvalue of the limit eigenvalue. For instance, T\'uran's inequality for the Legendre polynomials cannot be proved by our method. This is also the case for the recurrence from~\cite{MelczerMezzarobba2022,YuChen2022,BostanYurkevich2022a}.
    There are cases where the method seems to extend to situations where the matrices~$A(n)$ have a unique dominant eigenvalue, while the limit matrix does not. We plan to report on this in a future work.
\end{itemize}

\section*{Acknowledgements}
This work has been supported in part by the ANR project NuSCAP ANR-20-CE48-0014. It benefited from discussions with Alin Bostan and Mohab Safey El~Din.

\printbibliography
\end{document}